\documentclass{amsart}

\usepackage[T1]{fontenc}
\usepackage[margin=1in]{geometry}
\usepackage{amsmath,amssymb,amsthm}
\usepackage[colorlinks]{hyperref}
\usepackage[all]{hypcap}
\usepackage{thmtools}
\usepackage{thm-restate}
\usepackage[foot]{amsaddr}
\usepackage[numbers,sort]{natbib}
\usepackage{chngcntr}

\newcommand{\eq}[1]{\hyperref[eq:#1]{(\ref*{eq:#1})}}
\renewcommand{\sec}[1]{\hyperref[sec:#1]{Section~\ref*{sec:#1}}}
\newcommand{\app}[1]{\hyperref[app:#1]{Appendix~\ref*{app:#1}}}
\newcommand{\thm}[1]{\hyperref[thm:#1]{Theorem~\ref*{thm:#1}}}
\newcommand{\lem}[1]{\hyperref[lem:#1]{Lemma~\ref*{lem:#1}}}
\newcommand{\propo}[1]{\hyperref[prop:#1]{Proposition~\ref*{prop:#1}}}
\newcommand{\defn}[1]{\hyperref[defn:#1]{Definition~\ref*{defn:#1}}}
\newcommand{\ex}[1]{\hyperref[ex:#1]{Example~\ref*{ex:#1}}}
\newcommand{\fct}[1]{\hyperref[fct:#1]{Fact~\ref*{fct:#1}}}
\newcommand{\fig}[1]{\hyperref[fig:#1]{Figure~\ref*{fig:#1}}}
\newcommand{\figs}[1]{\hyperref[fig:#1]{Figures~\ref*{fig:#1}}}

\newcommand{\cas}[1]{\hyperref[case:#1]{Case~\ref*{case:#1}}}

\theoremstyle{plain}
\newtheorem{theorem}{Theorem}
\newtheorem{lemma}{Lemma}
\newtheorem{definition}{Definition}

\newtheorem{prop}{Proposition}

\theoremstyle{remark}
\newtheorem{case}{Case}

\counterwithin*{case}{theorem}
\counterwithin*{case}{lemma}
\counterwithin*{case}{prop}

\newcommand{\ket}[1]{|{#1}\rangle}

\newcommand{\bbra}[1]{\big\langle{#1}\big|}
\newcommand{\bket}[1]{\big|{#1}\big\rangle}

\newcommand{\CC}{\mathbb{C}}
\newcommand{\II}{\mathbb{I}}

\newcommand{\RR}{\mathbb{R}}

\newcommand{\clif}{\mathcal{C}}
\newcommand{\paul}{\mathcal{P}}
\newcommand{\XX}{\mathcal{X}}

\newcommand{\unitary}{\mathcal{U}}
\newcommand{\EE}{\mathcal{E}}

\DeclareMathOperator{\tr}{Tr}
\DeclareMathOperator{\poly}{poly}
\DeclareMathOperator{\modwithnospace}{mod}
\newcommand{\Mod}[1]{\ (\modwithnospace\ #1)}

\title{The Clifford group forms a unitary 3-design}
\author{Zak Webb$^{1}$} \email{zakwwebb@gmail.com}

\address{$^1$ Institute for Quantum Computing and Department of Physics \& Astronomy, University of Waterloo}

\begin{document}

\maketitle 
\date{}


\begin{abstract} Unitary $k$-designs are finite ensembles of unitary matrices that approximate the Haar distribution over unitary matrices. Several ensembles are known to be 2-designs, including the uniform distribution over the Clifford group, but no family of ensembles was previously known to form a 3-design.  We prove that the Clifford group is a 3-design, showing that it is a better approximation to Haar-random unitaries than previously expected.  Our proof strategy works for any distribution of unitaries satisfying a property we call Pauli 2-mixing and proceeds without the use of heavy mathematical machinery.  We also show that the Clifford group does not form a 4-design, thus characterizing how well random Clifford elements approximate Haar-random unitaries.  Additionally, we show that the generalized Clifford group for qudits is not a 3-design unless the dimension of the qudit is a power of 2.
\end{abstract}


\section{Introduction\label{sec:intro}}
The Clifford group is ubiquitous in quantum information.  Quantum error correction \cite{Pres98}, classical simulatibility \cite{Got99}, and randomized benchmarking \cite{EWSLC03} are only some of the areas in which the Clifford group plays a major role.  Because of this wide range of use, the Clifford group is often the second group examined when looking for a group of unitaries with some desired property (the first being the Pauli group).  

One such instance is in the attempt to approximate Haar-random unitaries.  A unitary $k$-design is an ensemble of unitaries such that the $k$-th moment of the ensemble is equal to the $k$-th moment of a Haar-random unitary.  It is easy to show that the uniform ensemble over the Pauli group forms a 1-design but that it does not form a 2-design \cite{RS09}.   In an attempt to find a better approximation to Haar-random unitaries, it was then shown that the uniform ensemble over the Clifford group forms a 2-design \cite{DLT02,DCEL09}.

Unitary $k$-designs are inspired from state $k$-designs, in which an ensemble of states is used to approximate the $k$-th moment of a Haar-random state.  Unitary designs are a generalization of state designs, as a state $k$-design can be derived from a unitary $k$-design by examining the orbit of any pure state under the elements of the unitary design.  State designs have found applications in state discrimination protocols \cite{AE07}, the phase lift algorithm \cite{GKK15}, and informationally complete measurements \cite{Scott06}.  Recently, Kueng and Gross showed that the stabilizer states (the orbit of any eigenvector of a Pauli element under the action of the Clifford group) form an exact state 3-design \cite{KG13,KG15}.  

The motivation behind approximating Haar-random unitaries arises from the simple fact that the symmetries of the Haar measure make analysis of protocols utilizing it simple, but the size of the measure makes it infeasible to sample from.  In order to sample from the Haar distribution with error $\epsilon$, we would need to place an $\epsilon$-net over the set of all unitaries, and a counting argument shows that some of these unitaries would require an exponential-length circuit over any fixed finite universal gate set.  The intuition behind $k$-designs is that we can perform an analysis assuming a Haar-random distribution, but claim that the same analysis holds for the design (assuming $k$ is large enough).  We then have a distribution that satisfies whatever property we were interested in, but is feasible to sample from.

Unitary 2-designs have found application in data hiding \cite{DLT02} and estimating the fidelity of channels \cite{DHCB05}, and thus effort has gone into their efficient construction with a near-linear time construction the best currently known \cite{DCEL09,CLLW15}.  For larger designs, unitary 3-designs have been used to show that quantum speed-ups occur for most unitaries \cite{BH13}, while unitary 4-designs have found application for state descrimination \cite{AE07} and relating state equilibriation to circuit length \cite{BCHHKM12,MRA13,BHH12}.  However, as the known applications of $k$-designs for $k>2$ are not as experimentally useful as $2$-designs, little successful effort has gone into finding examples of these exact designs.  The only previously known exact $k$-designs for $k>2$ were found numerically by analyzing the characters of finite groups \cite{GAE07,RS09}, and thus we have several examples of $3$-, $4$-, and $5$-designs on constant sized systems.  However, these methods do not generalize, and thus we do not currently know of any family of unitary $k$-designs for $k>2$.  This is in contrast with state designs, as the stabilizer states form 3-designs \cite{KG13,KG15}, as mentioned above.

This lack of exact $k$-designs has lead the community to the construction of approximate $k$-designs \cite{HL09}.  Related to quantum expanders \cite{DJNR09}, these approximate designs are defined so that the $k$-fold twirl over the ensemble is close to that of Haar-random unitaries in the completely bounded trace norm (also known as the diamond norm).  Further, it has been shown that random circuits of length $\poly(k)$ form approximate $k$-designs  \cite{HL09,BH13,BHH12}, giving an explicit (and simple) construction for these approximate designs.  As there exist approximate $k$-designs for arbitrary $k$, applications such as fooling small circuits, fast quantum equilibration, and the generation of topological order \cite{BHH12} are known, but these applications are usually interesting in the case of growing $k$.

In this paper, we show that the Clifford group forms an exact unitary 3-design.  We actually prove a slightly stronger statement: any ensemble of unitary matrices that is Pauli 2-mixing is a 3-design, where Pauli 2-mixing is a generalization of Pauli mixing to pairs of Pauli elements.  With a proof that the uniform ensemble of Clifford unitaries is Pauli 2-mixing, this shows that the Clifford group is a 3-design.  Our proof proceeds in a manner analogous to \cite{DLT02,CLLW15}, in that we examine the action of the $k$-fold Haar-random unitary twirl channel on some well-chosen operators.  We then examine the 3-fold twirl by a Pauli 2-mixing ensemble, and show that these two channels must be equal, which then proves that the ensemble is a 3-design.

We also show that no ensemble of Clifford unitaries forms a 4-design, and in particular that the result holds for the Clifford group.  This allows us to precisely characterize how well the uniform distribution over the Clifford group approximates a Haar-random unitary.  Further, this result shows that any search for exact 4-designs will require the examination of non-group sets of unitaries or groups of unitaries other than the Clifford group, since the only group on $n$ qubits properly containing the Clifford group is infinite.

After discovering that the uniform ensemble over the Clifford group is a 3-design, the author became aware of independent work by Huangjun Zhu \cite{Zhu15}, which proves a related result using representation theory.  In particular, Zhu's results show that the Clifford group is an exact 3-design, that no subgroup of the Clifford group is a 3-design (except on two qubits), and that the generalized Clifford group on qudits is not a 3-design (in addition to an application explaining why no discrete Wigner function is covariant with respect to the Clifford group).  Our work is complementary, in that our proof strategy works for arbitrary distributions of Clifford elements and the analysis of our proof gives intuition on why particular ensembles do not form $k$-designs.  

While examining Zhu's work, the author became interested in why our proof fails for qudits.  The discrepancy arises from the fact that there are $d$ possible commutation relations between two generalized Pauli operators, while there are only $2$ order-3 permutations in $S_3$.  This means that certain terms cannot be cancelled and our result does not hold for qudits.  Using this analysis, we give an alternative proof that the generalized Clifford group is not a 3-design.  (We actually prove that any ensemble of generalized Clifford operators is not a 3-design, in an analogous manner to how we prove that any ensemble of Clifford operators is not a 4-design.)

It should be noted that our proof that a Pauli 2-mixing ensemble forms a 3-design only requires an understanding of the Pauli group and the mathematical notation for operators and channels; no heavy mathematical machinery is required.  While a complete understanding of the $k$-fold Haar-random unitary twirl channel generally requires a use of Von Neumann's double commutant theorem or representation theory, our proof avoids this complete characterization.  (Our proof that the Clifford group does not form a 4-design and our proof that the generalized Clifford group does not form a 3-design both use Von Neumann's double commutant theorem to characterize the channel, but in an easily-understood manner.)

We now give a brief overview of the paper.  In \sec{math_prelim}, the requisite mathematical background and useful operators are defined and explained, including the definition of $k$-designs and Pauli 2-mixing.  In \sec{proof} we prove that any Pauli 2-mixing ensemble is a 3-design, which then proves that the Clifford group is a 3-design.  We show in \sec{not_4_design} that the Clifford group does not form a 4-design, and in \sec{gen_Pauli_not_design} that the generalized Clifford group is not a 3-design.  Finally, we discuss these results and their relation with previous works, along with some avenues of future research, in \sec{discussion}.  We give a self-contained proof that any Pauli mixing ensemble is a 2-design in \app{pauli_mixing_2_design}.


\section{Mathematical Preliminaries\label{sec:math_prelim}}

Let $\XX$ be a $d$-dimensional complex Euclidean space (e.g., $\XX = \CC^d$).  We are interested in the linear operators acting on this space, denoted $L(\XX)$.  Also of interest to us are the unitary operators acting on $\XX$, denoted by $\mathcal{U}(\XX)\subset L(\XX)$ and explicitly defined as $\mathcal{U}(\XX) = \{U\in L(\XX): U U^\dag = \II_\XX\}$, where $U^\dag$ is the Hermitian conjugate of $U$ and $\II_\XX$ is the identity operator on $\XX$.  Additionally, the completely-positive, trace-preserving, linear operators acting on $L(\XX)$, denoted by $C(\XX)$, are the quantum channels taking $\XX$ to $\XX$.

For any two operators $A,B\in L(\XX)$, their inner product is defined as $\langle A,B\rangle = \tr(A^\dag B)$.  Note that this notation is extremely similar to the group theoretic expression $\langle g_i\rangle_{i\in[m]}$ denoting the group generated by the set of generators $\{g_i\}_{i\in[m]}$.  We shall use both notations in our paper, but the context should enable the reader to differentiate between the two.

Additionally, if $\mathcal{S}$ is some set, we represent ordered tuples of elements from $\mathcal{S}$ by the bolded operator $\mathbf{s}$.  We then represent the $i$th component of the tuple $\mathbf{s}$ by $s_i$.  We will often define the tuple $\mathbf{s}$ by specifying the elements $s_i$ without explicitly defining $\mathbf{s}$.

Further, for any $n\in \mathbb{N}$, we will define the set $[n] := \{0,1,\cdots, n-1\}$, as this is a simple example of an ordered set of size $n$.


\subsection{Pauli and Clifford groups}

On a single qubit, the Pauli group $\widetilde{\paul}_1$ is a subset of $L\big(\CC^2\big)$ defined as
\begin{equation}
 \widetilde{\paul}_1 :=  \big\langle i \II_{\CC^2}, \sigma_X, \sigma_Z\big\rangle, \qquad \text{ where } \qquad 
 \sigma_X := \begin{pmatrix} 0 & 1\\ 1 & 0\end{pmatrix} \quad \text{ and }\quad 
 \sigma_Z := \begin{pmatrix} 1 & 0 \\ 0 & -1 \end{pmatrix}.
\end{equation}
We can construct the Pauli group on $n$ qubits $\widetilde{\paul}_n \subset L\big(\CC^{2^n}\big)$ by taking the $n$-fold tensor product of the Pauli group on a single qubit (i.e., $\widetilde{\paul}_n := \widetilde{\paul}_1^{\otimes n}$).

We are usually uninterested in the overall phase of an operator, so we define $\paul_n := \widetilde{\paul}_n/\langle i \II_{\CC^{2^n}}\rangle$, where as representatives we take $p\in \paul_n$ with $p^2 = \II_{\CC^{2^n}}$.  Further, we often want to only consider the non-identity Pauli elements and we define $\widehat{\paul}_n := \paul_n\setminus\big\{\II_{\CC^{2^n}}\big\}$.  It will also be useful to name the non-identity Pauli elements that square to the identity: $\overline{\paul}_n := \big\{ \pm p : p\in \widehat{\paul}_n\big\}$.

Part of the reason why the Pauli group is integral to quantum information is the simple fact that $\paul_n$ forms an orthogonal basis for $L\big(\CC^{2^n}\big)$.  Therefore, to completely characterize the effect of a channel $\Phi\in C\big(\CC^{2^n}\big)$, one need only understand the effect of $\Phi$ on each element $p\in \paul_n$.  This result greatly reduces the effort of characterizing channels, which is helpful both as a theoretical tool and as an experimental procedure.  Along these lines, we can expand any operator $A\in L\big(\CC^{2^n}\big)$ in the Pauli basis as
\begin{equation}
  A = \frac{1}{2^n}\sum_{p\in\paul_n}  \langle p,A\rangle p,
\end{equation}
where the factor $2^{-n}$ arises from the normalization of the Pauli elements.

A useful fact about the Pauli group is that each $p\in\widehat{\paul}_n$ commutes with exactly half of the elements of $\paul_n$, and anti-commutes with the other half.  This leads us to define (for $p_1,p_2\in \widetilde{\paul}_n$)
\begin{equation}
  F(p_1,p_2) := \begin{cases} 0 & p_1p_2 = p_2 p_1\\
    1 & p_1p_2 = -p_2 p_1.\end{cases}
\end{equation}

Note that this previous discussion assumes we are working with qubits.  The Pauli group can be generalized to $d$-dimensional systems, with many of the same properties.  In particular, if $\big\{\ket{j}\big\}_{j\in[d]}$ is an orthonormal basis for the space $\CC^{d}$, then the generalized Pauli group is a subset of $L\big(\CC^{d}\big)$ defined as
\begin{equation}
  \widetilde{\paul}_1^d := \big\langle \tilde{\omega} \II_{\CC^d}, \sigma_{X}^d, \sigma_{Z}^d\big\rangle,
      \quad \sigma_{X}^d \ket{j} := \ket{j+1\Mod d},
      \quad \sigma_{Z}^d \ket{j} := \omega^j \ket{j },
      \quad \omega := e^{2\pi i /d},
      \quad \tilde{\omega} := \begin{cases} \omega & d\text{ odd}\\ e^{\pi i /d} & d\text{ even}.\end{cases}
    \end{equation}
We define $\tilde{\omega}$ differently for even and odd dimension as $(\sigma_{X}^d\sigma_{Z}^d)^d$ equals either $-\II$ or $\II$ depending on whether $d$ is even, and thus the allowed global phases in the group differ depending on the parity of $d$.  With this, we can then define the generalized Pauli group on $n$ qudits as $\widetilde{\paul}_n^d := \big(\widetilde{\paul}^d_1\big)^{\otimes n}$, as well as the group without phases $\paul_n^d := \widetilde{\paul}_n^d / \langle \tilde{\omega} \II_{\CC^{d^n}}\rangle$ and the group without the identity element $\widehat{\paul}_n^d = \paul_n^d \setminus \{\II_{\CC^{d^n}}\}$.  Additionally, it will be useful to name the non-identity elements whose order divides $d$:  $\overline{\paul}_n^d := \big\{ \omega^\ell p : p\in \widehat{\paul}_n^d \text{ and } \ell\in [d]\big\}$.

\begin{table}
  \caption{Definitions for sets related to the Pauli group.}  
  \begin{tabular}{|r@{\;$:=$\;}ll|}
    \hline
    $\widetilde{\paul}_1^d$ &$\langle \tilde{\omega} \II_{\CC^d},\sigma_X^d,\sigma_Z^d\rangle$ & The Pauli group on one qudit\\
    $\widetilde{\paul}_n^d$ & $(\widetilde{\paul}_1^d)^{\otimes n}$ &The Pauli group on $n$ qudits\\
    $\paul_n^d$ &$ \widetilde{\paul}_n^d / \langle\tilde{\omega} \II_{\CC^{d^n}}\rangle $ & The Pauli group modulo phases\\
    $\widehat{\paul}_n^d$ & $\paul_n^d \setminus \{\II_{\CC^{d^n}}\} $& The Pauli group without phases or the identity element\\
    $\overline{\paul}_n^d$ &$\big\{\omega^\ell p : p\in \widehat{\paul}_n^d, l\in[d]\big\}$ &Non-identity Pauli elements with order that divides $d$\\\hline
  \end{tabular}
\end{table}

The generalized Pauli group forms a basis for $L\big(\CC^{d^n}\big)$, much like the Pauli group.  One difference between the two groups is that the commutation relations of the generalized Pauli group are not two-valued, in that $p_1p_2= \omega^j p_2p_1$ for some $j\in [d]$.  This forces us to extend the definition of $F(p_1,p_2)$ to be the $d$-valued function equal to the power of $\omega$ in the commutation relation. (Note that for $d$-dimensional systems it is no longer symmetric in its arguments).  However, much like how each non-identity Pauli element commutes and anti-commutes with exactly half of the Pauli group, each non-identity element of the generalized Pauli group has a commutation relation $\ell$ with a $1/d$ fraction of the generalized Pauli group. 

With these definitions and facts about the Pauli group, we can define the Clifford group as the normalizer of the Pauli group in $\mathcal{U}\big(\CC^{2^n}\big)$.  More concretely, we have
\begin{equation}
  \widetilde{\clif}_n := \big\{c \in \mathcal{U}\big(\CC^{2^n}\big) : \forall p\in \paul_n, cpc^\dag \in  \widetilde{\paul}_n\big\}.
\end{equation}
Noting that $c\in \widetilde{\clif}_n$ implies $e^{i\theta}c\in \widetilde{\clif}_n$ for all $\theta\in\RR$, we define $\clif_n := \widetilde{\clif}_n/ \big\{e^{i\theta}\II_{\CC^{2^n}} : \theta\in \RR\big\}$.

In analogy to the Clifford group, we can define the generalized Clifford group to be the normalizer of the generalized Pauli group.  Explicitly,  
\begin{equation}
  \widetilde{\clif}_n^d := \big\{ c \in \unitary\big(\CC^{d^n}\big) : \forall p \in \paul_{n}^d, cpc^\dag \in \widetilde{\paul}_n^d\big\}
   \qquad \text{ and }\qquad   
   \clif_n^d := \widetilde{\clif}_n^d/\big\{e^{i\theta} \II_{\CC^{d^n}} : \theta \in \RR\big\}.
\end{equation}

We will often be interested in the subsets of $\clif_n^d$ in which a particular element of $\paul_n^d$ gets mapped to another.  As such, for $p\in \widehat{\paul}_n^d$ and $q\in \overline{\paul}_n^d$ let us define the set
\begin{equation}
  \clif_{p\rightarrow q} := \{c\in\clif_n^d : cpc^{\dag} = q\}
\end{equation}
consisting of exactly those generalized Clifford elements that map $p$ to $q$ under conjugation.  We need only examine $q\in \overline{\paul}_n^d$, as the representative $p\in \widehat{\paul}_n^d$ satisfy $p^d =\II_{\CC^{d^n}}$.

 We can generalize this idea to Clifford elements that have a fixed action on sets of Pauli elements.  In particular let $\mathbf{p}\in\big(\widehat{\paul}_n^d\big)^{ m}$ and  $\mathbf{q}\in \big(\overline{\paul}_n^d\big)^{ m}$, and define
\begin{equation}
  \clif_{\mathbf{p}\rightarrow \mathbf{q}} = \{c\in \clif_n^d : \forall i\in [m], cp_ic^\dag =  q_i\}.
\end{equation}
Note that the set $\clif_{\mathbf{p}\rightarrow \mathbf{q}}$ will be empty if the commutation relations between the elements of $\mathbf{p}$ and $\mathbf{q}$ are not consistent.  

For a more in depth examination of the Pauli and Clifford groups, including an explanation of their symplectic representation, see \cite{CRSS96,Far14,Got99}.  


\subsection{Permutation operators}

When analyzing unitary $k$-designs on $\XX = \CC^d$, we will be attempting to understand operators from $L\big(\XX^{\otimes k}\big)$.  It turns out that the operators corresponding to a permutation of these underlying spaces are useful in the analysis of $k$-designs.

Let $\big\{\ket{i}\big\}_{i\in[d]}$ be an orthonormal basis for $\XX$.  For each element $\pi\in S_k$ (where $S_k$ is the symmetric group on $k$ elements), we can define the operator $W_\pi \in L\big(\XX^{\otimes k}\big)$ that permutes the subsystems $\XX_i$ according to $\pi$:
\begin{equation}
  W_\pi = \sum_{\forall j\in [k], i_j \in[d]} \bket{\pi(i_{1},i_{2},\cdots,i_{k})}\bbra{i_{1},i_{2},\cdots,i_{k}}.
\end{equation}

As we are interested in the case $k=3$, let us expand $W_{(123)}$ in the generalized Pauli basis as
\begin{equation}
  W_{(123)} = \sum_{p_1,p_2,p_3\in \paul^d} \alpha(p_1,p_2,p_3) p_1\otimes p_2\otimes p_3,
\end{equation}
where
\begin{align}
  \alpha(p_1,p_2,p_3) & = \frac{1}{d^3}\tr\big[W_{(123)} \big(p_1^\dag \otimes p_2^\dag \otimes p_3^\dag\big)\big]
  = \frac{1}{d^3}\tr\big[p_1^\dag p_3^\dag p_2^\dag \big]
  = \frac{1}{d^2}\begin{cases}
        \beta  & p_3 = \beta (p_1p_2)^\dag\\
         0 & \text{otherwise}.
        \end{cases}
\end{align}
If we then take $p_3 = (p_1p_2)^\dag$ (i.e, $\beta = 1$), this yields
\begin{equation}
  W_{(123)} = \frac{1}{d^2}\sum_{p_1,p_2\in\paul^d} p_1\otimes p_2\otimes p_2^\dag p_1^\dag.
\end{equation}
If we assume that $d = 2^n$ for some $n$ (i.e., we are working with $n$ qubits), then we can take $p^\dag = p$ and the equation simplifies slightly:
\begin{equation}
  W_{(123)} = \frac{1}{2^{2n}} \sum_{p_1,p_2\in\paul_n} p_1\otimes p_2\otimes p_2 p_1.
\end{equation}

In a similar manner, we can see that
\begin{align}
  W_{(12)}  &= \frac{1}{d} \sum_{p_1\in\paul^d} p_1 \otimes p_1^\dag \otimes \II_\XX, &W_{(13)} &=\frac{1}{d} \sum_{p_1\in\paul^d} p_1 \otimes \II_\XX \otimes p_1^\dag,\\
  W_{(23)} &= \frac{1}{d} \sum_{p_1\in\paul^d}\II_\XX \otimes  p_1 \otimes p_1^\dag,  & W_{(321)} &= \frac{1}{d^2} \sum_{p_1,p_2\in\paul^d} p_1 \otimes p_2 \otimes p_1^\dag p_2^\dag.
\end{align}

We will also be interested in the case $k=4$ for qubit systems, and in particular we will be interested in those $W_\pi$ that have nonzero support on operators in $\widehat{\paul}_n^{\otimes 4}$ (i.e, those $W_\pi$ that satisfy $\langle X, W_\pi\rangle \neq 0$ for some $X\in \widehat{\paul}_n^{\otimes 4}$).  The only operators with this property are the rank-4 permutations and the commuting two-cycle permutations, as they are the only permutations that have no fixed point.   It can be shown in a manner identical to that of $W_{(123)}$ that
\begin{align}
  W_{(1234)} &= \frac{1}{2^{3n}}\sum_{p_1,p_2,p_3\in \paul_n} p_1\otimes p_2\otimes p_3 \otimes p_3p_2p_1 ,
  &W_{(12)(34)} &= \frac{1}{2^{2n}} \sum_{p_1,p_2\in \paul_n} p_1\otimes p_1 \otimes p_2\otimes p_2,
\end{align}
with an appropriate rearrangement of the subscripts for the permutations that are conjugate to these.

There is a nice relation between operators on $L\big(\XX^{\otimes k}\big)$ that commute with $X^{\otimes k}$ for all $X\in L(\XX)$ and these permutation operators.  The following theorem will be used to characterize the $k$-fold Haar-random unitary twirl channel:  

\begin{theorem}[Theorem 7.16 of \cite{Watrous}]
  Let $k$ be a positive integer, let $\XX = \CC^d$, and let $X \in L\big(\XX^{\otimes k}\big)$ be an operator.  The following statements are equivalent:
  \begin{enumerate}
    \item It holds that $\big[X,Y^{\otimes k}\big] = 0$ for all choices of $Y\in L(\XX)$.
    \item It holds that $\big[X,U^{\otimes k}\big] = 0$ for all choices of $U\in \mathcal{U}(\XX)$.
    \item For some choice of vector $u \in \CC^{|S_k|}$, it holds that 
    \begin{equation}
      X = \sum_{\pi\in S_k} u_\pi W_\pi.
    \end{equation}    
  \end{enumerate}
\label{thm:watrous}
\end{theorem}

The proof of this theorem applies Von Neumann's double commutant theorem, and as such hides some heavy mathematical machinery behind the theorem statement.  We will not need this theorem to show that the Clifford group is a 3-design, but it will be useful to us in order to show that certain ensembles do not form $k$-designs.


\subsection{The Haar measure on $\mathcal{U}(\XX)$.}

In an attempt to generalize the idea that the size of an interval on $\RR$ does not depend on its location, the Haar measure on a group is defined so that the measure is invariant under the action of the group.  In the cases we are interested in, this means that an arbitrary rotation of a set of unitaries does not change the set's size. 

Haar-random unitaries have found several uses in quantum information.  The symmetries of the measure ease the analysis of many systems, and the averaging properties allow for the construction of an ``average'' state.  Several theoretical results in quantum information have simple proofs using Haar-random unitaries, such as \cite{HLSW04,HHWY08,HW08,HLW06}.

Let us now define the Haar measure.  For $\XX = \CC^d$, we are interested in a measure $\eta$ on $\mathcal{U}(\XX)$ that satisfies
\begin{equation}
  \int_{\mathcal{U}(\XX)} d\eta(U) = 1,
\end{equation}
as well as satisfying the property that for any element $U\in \mathcal{U}(\XX)$ and any (measurable) subset $S\subset \mathcal{U}(\XX)$, 
\begin{equation}
  \int_{S} d\eta(V) = \int_{US} d \eta (V) = \int_{S} d\eta(U V) = \int_S d\eta(VU).
\end{equation}
These two requirements uniquely determine the probability measure $\eta$, which is called the Haar measure. 

For our purposes we need only note that  when integrating according to the Haar measure over $\unitary(\XX)$, multiplying the variable of integration by a unitary does not change the value of the integrand.  This ``unitary invariance of the Haar measure'' will allow us to prove some useful equalities. 

For more information on the Haar measure, including a formal derivation, see Chapter 7 of \cite{Watrous}.


\subsection{$k$-designs}

The motivation behind unitary designs is to approximate the distribution of Haar-random unitaries without requiring an infeasible amount of resources.  Namely, a design should be a finite (and hopefully small) ensemble such that a unitary sampled from the design is hard to distinguish from a unitary sampled from the Haar measure.  Keeping with this idea, one definition of $k$-designs requires that the first $k$-moments of the two distributions are equal, so that at least $(k+1)$ applications of the sampled unitary are needed in order to distinguish between the distributions.  While several equivalent definitions exist for $k$-designs (see \cite{Low10} for a proof of equivalence), we shall focus on one defined in terms of the $k$-fold unitary twirl.  

A (finite) ensemble of unitaries over $\XX$ is denoted by $\EE = \big\{(\alpha_i,U_i)\big\}_{i\in[m]}$, where each $U_i$ is a unitary over $\XX$, $\alpha \in \RR^{m}$ is a probability vector, and we think of $\alpha_i$ as the probability that $U_i$ is chosen from the ensemble.  If we refer to a finite set $\mathcal{S}\subset \mathcal{U}(\XX)$ as an ensemble, then we are referring to the ensemble where each element of $\mathcal{S}$ is chosen with equal probability. 

For any complex Euclidean space $\mathcal{X}$ and any positive integer $k$, let us define the channel $T_k\in C\big(\mathcal{X}^{\otimes k}\big)$ to be the $k$-fold twirl by Haar-random unitaries with an action on $X\in L\big(\XX^{\otimes k}\big)$ given by
\begin{equation}
  T_k(X) := \int d\eta(U) U^{\otimes k} X \big(U^\dag\big)^{\otimes k}.
\end{equation}
Similarly, we can define the $k$-fold twirl by any (finite) ensemble of unitaries $\EE$ over $\XX$ to be the channel $\Psi_{\EE,k} \in C\big(\mathcal{X}^{\otimes k}\big)$ with an action on $X\in L\big(\XX^{\otimes k}\big)$ given by
\begin{equation}
  \Psi_{\mathcal{E},k}(X) := \sum_{(\alpha,U)\in\EE} \alpha \;U^{\otimes k} X \big(U^\dag\big)^{\otimes k}.
\end{equation}

\begin{definition}[$k$-design]  Let $\mathcal{X}$ be a complex Euclidean space and $k$ be a positive integer.  A finite ensemble $\mathcal{E}$ of unitaries over $\XX$ is a unitary $k$-design if and only if $\Psi_{\EE,k}(X) = T_k(X)$, for all $X\in L\big(\mathcal{X}^{\otimes k}\big)$.
\end{definition}

(Note that this definition of $k$-designs is sometimes called a weighted $k$-design.  An unweighted $k$-design is a design where each probability $\alpha$ is equal.)

Several families of exact 2-designs are known, including the Clifford group and some of its subgroups.  Additionally, several $k$-designs on small systems have been found by analyzing the characters of finite groups \cite{GAE07,RS09}.  However, no construction of $3$-designs on $n$ qubits was previously known.  (Approximate $k$-designs are known for arbitrary $k$: simply take a random circuit of length $\poly(k)$ \cite{BHH12}.)

Much is known about the properties of a $k$-design, even if no construction currently exists.  In particular, it is known that an exact $k$-design is also a $(k-1)$-design, there are bounds on the minimal size of a $k$-design for a $d$-dimensional system, and there exists a characterization of $k$-designs in terms of a simple function of the design elements \cite{GAE07,RS09,Low10}.  

An additional property that will be of use to us is the following relation between any $k$-fold unitary twirl and $T_k(X)$, with the intuition that the channel $T_k$ acts as a projector and the $k$-fold twirl by an ensemble is only an approximation to $T_k$.
\begin{prop}  Let $\mathcal{X}$ be a complex Euclidean space, $\EE$ be a finite ensemble of unitaries over $\XX$, and $k$ be a positive integer.  For all $X\in L(\mathcal{X}^{\otimes k})$,  $T_k$ satisfies $T_k(X) = T_k\big(\Psi_{\EE,k}(X)\big)$ .
\label{prop:unitary_twirl_is_projective}
\end{prop}
\begin{proof}
  This makes use of the unitary invariance of the Haar measure.  Namely
  \begin{align}
    T_k(X) &= \int d\eta(U) U^{\otimes k} X \big(U^\dag\big)^{\otimes k} 
    = \sum_{(\alpha, V)\in \EE} \alpha \int d\eta(UV) (UV)^{\otimes k} X \big((UV)^\dag\big)^{\otimes k} \\
    &= \int d\eta(U) U^{\otimes k} \Bigg[\sum_{(\alpha,V)\in \EE} \alpha V^{\otimes k} X \big(V^\dag\big)^{\otimes k}\Bigg] \big(U^\dag\big)^{\otimes k} 
    = T_k\big(\Phi_{\EE,k}(X)\big).\qedhere
  \end{align}
\end{proof}

With this relation between the $k$-fold twirl by Haar-random unitaries and an arbitrary ensemble, it will be useful to know when $T_k(X) = X$.  This is where the permutation operators become useful:
\begin{prop}
  Let $\XX$ be a complex Euclidean space, $k$ be a positive integer, and $\pi \in S_k$.  The permutation operator $W_\pi \in L\big(\XX^{\otimes k}\big)$ satisfies $T_k(W_\pi) = W_\pi .$
\label{prop:unit_twirl_preserves_permutation}
\end{prop}
\begin{proof}
  First note that for any $V\in \unitary(\XX)$ and any $\pi\in S_k$, $\big[W_\pi,V^{\otimes k}\big] = 0$ from the simple fact that each register of $V^{\otimes k}$ contains the same operator.  Using this, we can see
  \begin{equation}
    T_k(W_\pi) = \int d\eta(U) U^{\otimes{k}} W_\pi \big(U^\dag\big)^{\otimes k} = W_\pi \int d\eta(U) \big(UU^\dag\big)^{\otimes k} = W_\pi. \qedhere
  \end{equation}
\end{proof}

These two propositions give us enough insight to understand the action of $T_k$ and give a sufficient condition for an ensemble to form a $k$-design.  In particular, if we can show that the image of a $k$-fold twirl over some ensemble $\EE$ is contained within the span of the $W_\pi$, then $\EE$ forms a $k$-design.  Note that this proof strategy does not require us to explicitly know the action of $T_k$ on an arbitrary operator $X$.

We will, however, require a characterization of the $k$-fold Haar-random twirl channel in our proofs that certain ensembles do not form $k$-designs.  Specifically, we need to show that the image of any $X$ under $T_k$ is contained within the span of the $W_\pi$. If we then show that the image of some channel is not contained within the span of the $W_\pi$, this will allow us to claim that the channel is not equal to $T_k$.

\begin{prop} Let $\XX$ be a complex Euclidean space and $k$ be a positive integer.  For any $X\in L\big(\XX^{\otimes k}\big)$,
\begin{equation}
  T_k(X) = \sum_{\pi \in S_k} \alpha_\pi(X) W_\pi,
\end{equation}
where $W_\pi \in L(\XX^{\otimes k})$ is the permutation operator corresponding to $\pi$ and $\alpha(X)\in \CC^{|S_k|}$.   
\label{prop:unit_twirl_spanned_by_permutation}
\end{prop}
\begin{proof}
  Let $V\in \mathcal{U}(\XX)$.  Note that
  \begin{equation}
    V^{\otimes k} T_k(X)\big(V^\dag\big)^{\otimes k} = \int d\eta(U) (VU)^{\otimes k} X\big(U^\dag V^\dag\big)^{\otimes k} = \int d\eta(U) U^{\otimes k} X \big(U^\dag\big)^{\otimes k} = T_k(X),
  \end{equation}
  where we have used the unitary invariance of the Haar measure.  This gives us that $\big[T_k(X), V^{\otimes k}\big] = 0$, and the result follows from an application of \thm{watrous}.
\end{proof}
Though our proof does not make explicit use of representation theory or any deep mathematical theorems, an application of Von Neumann's double commutant theorem is hidden in the proof of \thm{watrous}.  As such, any result using the characterization of \propo{unit_twirl_spanned_by_permutation} also uses some heavy mathematical machinery.


\subsection{Pauli Mixing}

A key idea in our proof that the Clifford group is a unitary 3-design is the understanding that a random element of the Clifford group sends any pair of Pauli elements to a uniformly random pair, thus breaking any correlation between Pauli elements.  This idea of uniform mixing over the Pauli group is enough to prove that an ensemble forms a 3-design.  To make these ideas more concrete, we require some additional definitions.

We first want to work with ensembles that remain unchanged under transformations by a Pauli element:
\begin{definition}[Pauli-invariant]  If $\XX = \CC^{d^n}$, then an ensemble $\EE$ is (right) \emph{Pauli-invariant} if for every $p\in \paul_n^d$, 
\begin{equation}
  (\alpha, U)\in \EE \quad \Rightarrow \quad (\alpha, e^{i \theta_U}U p)\in \EE.
\end{equation}
for some $\theta_U \in \RR$.
\end{definition}

   Additionally, we will often be interested in ensembles of Clifford elements, so that we can define the sub-ensemble of elements that have a particular action on a given Pauli element.  Namely, if $\EE$ is an ensemble of unitaries over $\CC^{d^n}$ consisting only of generalized Clifford elements (a $\clif_n^d$-ensemble), and if $p\in \widehat{\paul}_n^d$ and $q\in \overline{\paul}_n^d$, then $\EE_{p\rightarrow q}$ is the sub-ensemble (ensemble of unitaries with a probability vector that sums to less than 1) that only contains the generalized Clifford elements that take $p$ to $q$ under conjugation:
 \begin{equation}
   \EE_{p\rightarrow q} := \big\{(\alpha,U) \in \EE : U \in \clif_{p\rightarrow q}\big\}.
 \end{equation}
 We can extend this to sets of generalized Pauli elements: if  $\mathbf{p}\in\big( \widehat{\paul}_n^d\big)^{m}$ and $\mathbf{q}\in \big(\overline{\paul}_n^d\big)^{m}$ we define
 \begin{equation}
   \EE_{\mathbf{p}\rightarrow \mathbf{q}} := \big\{(\alpha,U)\in\EE : U\in \clif_{\mathbf{p}\rightarrow \mathbf{q}}\big\}.
 \end{equation} 
 Note that $\EE_{\mathbf{p}\rightarrow \mathbf{q}}$ will be empty if no element of $\EE$ satisfies the required constraints.

At this point we can define the notion of Pauli mixing in our framework.  Pauli mixing captures the intuitive idea that a distribution of Clifford elements uniformly permutes a given Pauli element.

\begin{definition}[Pauli Mixing]
Let $\EE$ be a Pauli-invariant $\clif_n^d$-ensemble.  We  say that $\EE$ is \emph{Pauli Mixing} if for all $p\in \widehat{\paul}_n^d$ the action of $\EE$ on $p$ is to map it to a uniform distribution over $\overline{\paul}_n^d$.  Explicitly, $\EE$ is Pauli mixing if for all $p\in \widehat{\paul}_n^d$ and $q\in \overline{\paul}_n^d$,
\begin{equation}
  \sum_{(\alpha,U)\in \EE_{p\rightarrow q}} \alpha= \frac{1}{\big| \overline{\paul}_n^d\big|} = \frac{1}{d \big|\widehat{\paul}_n^d\big| } = \frac{1}{d(d^{2n}-1)}.
\end{equation}
\label{def:pauli_mixing}
\end{definition}

The idea of Pauli mixing was used implicitly by \cite{DLT02}, and was initially defined by Gross, Audenaert, and Eisert in \cite{GAE07}, who showed that Pauli mixing is sufficient to form a 2-design.  Cleve, Leung, Liu, and Wang \cite{CLLW15} then generalized this result to ensembles.

\begin{restatable}{lemma}{pauliMixingTwoDesign} If $\EE$ is a Pauli-invariant and Pauli mixing $\clif_n^d$-ensemble, then $\EE$ is a 2-design.
\label{lem:pauli_mixing_2_design}
\end{restatable}
As our framework is slightly different from those of previous papers, we include a proof in \app{pauli_mixing_2_design}.

While Pauli mixing is sufficient to show that an ensemble is a 2-design, it is not a sufficient condition to form a 3-design.  We must then generalize this property to pairs of Pauli elements.  While one might at first imagine that we would want a mixing procedure to uniformly mix over all pairs of Pauli operators, this does not take into account the commutation relations of the Pauli group.  It turns out that the mixing condition will differ depending on whether the initial two operators commute.  Hence, let us define
\begin{align}
  H_{\text{commute}} &= H_0 := \big\{ (q_1,q_2) \in \big(\overline{\paul}_n\big)^2 :q_1\not\propto q_2 \text{ and } F(q_1,q_2) = 0\big\}\\
  H_{\text{anti-commute} }&= H_1 := \big\{ (q_1,q_2)\in \big(\overline{\paul}_n\big)^2: q_1\not\propto q_2 \text{ and }F(q_1,q_2) = 1\big\}.
\end{align}
If we are considering qudits, we can define similar subsets:
\begin{equation}
  H_{\ell} := \big\{ (q_1,q_2)\in \big(\overline{\paul}_n^d\big)^2 : q_1\not\propto q_2 \text{ and } F(q_1,q_2) = \ell\big\}.
\end{equation}

Note that to determine the sizes of these groups, we are free to choose the first element so long as it differs from the identity, and thus there are $d |\widehat{\paul}_{n}^d|$ such elements (after allowing for the $d$ different phases).  For the second element, there are $|\paul_n^d|/d - 2$ elements that commute with the first not including phases, and thus the total size of $H_{0}$ is
\begin{align}
  |H_{0}| &= d^2 |\widehat{\paul}_n^d | \bigg(\frac{|\paul_n^d|}{d}- 2\bigg) = d |\widehat{\paul}_n^d| \big(| \widehat{\paul}_n^d| +1 - 2d\big).
\end{align}
If we instead want an element with a given commutation relation, note that there are exactly $|\paul_n^d|/d$ such elements not including phases, and thus we have that
\begin{align}
  |H_{\ell}| &= d^2 |\widehat{\paul}_n^d| \bigg(\frac{\paul_n^d}{d} \bigg) = d |\widehat{\paul}_n^d| \big(|\widehat{\paul}_n^d| +1 \big).
\end{align}

The intuitive idea behind Pauli 2-mixing is that the ensemble $\EE$ sends a given pair of Pauli elements to a uniform mixture over all pairs that have the same commutation relation as the original pair.

\begin{definition}[Pauli 2-mixing]
  Let $\EE$ be a $\clif_n^d$-ensemble.  We say that $\EE$ is \emph{Pauli 2-mixing} if for all $p_1,p_2\in \widehat{\paul}_n^d$ and $p_1\not \neq p_2$, the action of $\EE$ is to map $(p_1,p_2)$ to a uniform distribution over $H_{F(p_1,p_2)}$.  Explicitly, $\EE$ is Pauli 2-mixing if for all $\mathbf{p}=(p_1,p_2)\in \big(\widehat{\paul}_n^d\big)^2$ with $p_1\neq p_2$, and all $\mathbf{q}\in H_{F(p_1,p_2)}$,
 \begin{equation}
   \sum_{(\alpha,U)\in \EE_{\mathbf{p}\rightarrow \mathbf{q}}} \alpha = \frac{1}{| H_{F(p_1,p_2)}|} 
     = \begin{cases} 
       \frac{1}{d \big|\widehat{\paul}_n^d\big| \big( 
       \big|\widehat{\paul}_n^d\big| + 1 -2d\big)}
       =  \frac{1}{d (d^{2n}-1)(d^{2n} - 2d)}  & F(p_1,p_2) = 0\\
       \frac{1}{d\big|\widehat{\paul}_n^d\big| \big(\big|\widehat{\paul}_n^d\big|+1\big)} 
       =\frac{1}{d^{2n + 1}(d^{2n}-1)} &  F(p_1,p_2) \neq 0 .\end{cases} 
 \end{equation}
\end{definition}

Note that if an ensemble is Pauli 2-mixing, then it is also Pauli mixing:
\begin{lemma}
  If $\EE$ is a Pauli 2-mixing $\clif_n^d$-ensemble, then $\EE$ is Pauli mixing.
\label{lem:2_mixing_implies_mixing}
\end{lemma}
\begin{proof}
  Given $p_1\in \widehat{\paul}_n^d$, we want to analyze those generalized Clifford elements that have a specific action on $p_1$, but the definition of Pauli 2-mixing only gives us information of the action on pairs of Pauli elements.  As such, let $p_2 \in \widehat{\paul}_n^d$ be such that $F(p_1,p_2) = 1$ and define $\mathbf{p} = (p_1,p_2)$.  (We only need $p_2$ to be some element of $\widehat{\paul}_n^d$ that differs from $p_1$, and this is the only element we can guarantee exists for all $n$ and $d$.)  
  
  For any $q_1\in \overline{\paul}_n^d$, if the action of a generalized Clifford element is to map $p_1$ to $q_1$, then the action of the element on $p_2$ must be to send it to a $q_2\in \overline{\paul}_n^d$ such that $(q_1,q_2)\in H_1$, since the generalized Clifford group preserves generalized Pauli elements and any conjugation map preserves commutation relations.  
  
 As such, let $q_1\in \overline{\paul}_n^d$.  We can use the above reasoning to decompose the set $\EE_{p_1\rightarrow q_1}$ by the action of the generalized Clifford elements on $p_2$.  In particular, we decompose $\EE_{p_1\rightarrow q_1}$ into sets $\EE_{(p_1,p_2)\rightarrow (q_1,q_2)}$ where $(q_1,q_2)\in H_1$.  Explicitly, we have:
  \begin{equation}
    \sum_{(\alpha,U)\in \EE_{p_1\rightarrow q_1}} \alpha 
    = \sum_{\substack{q_2\in \overline{\paul}_n^d:\\
        \mathbf{q}:=(q_1,q_2)\in H_{1}}} \sum_{(\alpha,U)\in \EE_{\mathbf{p}\rightarrow \mathbf{q}}} \alpha 
    = \sum_{\substack{q_2\in \overline{\paul}_n^d:\\
         (q_1,q_2)\in H_{1}}} \frac{1}{d^{2n+1}(d^{2n}-1)} 
    = \frac{1}{d(d^{2n}-1)},
  \end{equation}
  where we used the decomposition of $\EE_{p_1\rightarrow q_1}$ in the first equality, the fact that $\EE$ is Pauli 2-mixing in the second equality, and the fact that $F(q_1, q_2) = 1$ for $d^{2n}$ elements $q_2\in\overline{\paul}_n^d$ in the third equality.  Hence, $\EE$ is Pauli mixing.
\end{proof}

We will use this lemma to show that any ensemble $\EE$ on $n$ qubits that is Pauli-invariant and Pauli 2-mixing is a 3-design.  With a proof that the uniform ensemble over the generalized Clifford group is Pauli-invariant and Pauli 2-mixing, this will prove that the Clifford group is a 3-design.

\begin{lemma}
The uniform ensemble $\EE$ over the generalized Clifford group on $n$ qudits is Pauli-invariant and Pauli 2-mixing.
\label{lem:cliffords_are_2_mixing}
\end{lemma}
\begin{proof}
  As the generalized Pauli group is a subgroup of the generalized Clifford group, we have that $\EE$ is Pauli-invariant.  Hence, we need only show that it is Pauli 2-mixing.
  
  Given $\mathbf{p}= (p_1,p_2)\in \big( \widehat{\paul}_{n}^d\big)^2$ such that $p_1\neq p_2$, let $\mathbf{q},\mathbf{r} \in H_{F(p_1,p_2)}$.  Since the components of $\mathbf{q}$ and $\mathbf{r}$ have the same commutation relation, there exists some generalized Clifford element $c_0$ that maps the components of $\mathbf{q}$ to the components of $\mathbf{r}$ (e.g., $c_0q_ic_0^\dag =r_i$ for $i\in\{1,2\}$).  We can then see that for all $c\in \clif_{\mathbf{p}\rightarrow \mathbf{q}}$, the element $c_0c$ satisfies $ c_0cp_ic^\dag c_0^\dag = c_0 q_i c_0^\dag =r_i,$ implying that $c_0c\in \clif_{\mathbf{p}\rightarrow \mathbf{r}}$ and $|\clif_{\mathbf{p}\rightarrow \mathbf{q}}| \leq |\clif_{\mathbf{p}\rightarrow \mathbf{r}}|$.  As these two elements were arbitrarily chosen, this implies that for all $\mathbf{q}\in H_{F(p_1,p_2)}$, $|\clif_{\mathbf{p}\rightarrow \mathbf{q}}| = |\clif_{\mathbf{p}\rightarrow \mathbf{p}}|$ (where we note that $\mathbf{p}\in H_{F(p_1,p_2)}$).
  
  Since each generalized Clifford element maps $p_1$ and $p_2$ to generalized Pauli elements that have the same commutation relations, we have
  \begin{equation}
    |\clif_n^d| = \sum_{\mathbf{q}\in H_{F(p_1,p_2)}} |\clif_{\mathbf{p}\rightarrow \mathbf{q}}| = |\clif_{\mathbf{p}\rightarrow \mathbf{p}}|  |H_{F(p_1,p_2)}|.
  \end{equation}
  Rearranging the terms then gives 
  \begin{equation}
    |\clif_{\mathbf{p}\rightarrow \mathbf{p}}| 
      = \frac{|\clif_n^d|}{|H_{F(p_1,p_2)}|}
      = |\clif_n^d| \begin{cases}\frac{1}{d(d^{2n}-1) (d^{2n}  -2d)} & F(p_1,p_2) = 0\\
       \frac{1}{d^{2n+1}(d^{2n}-1)} & F(p_1,p_2) \neq 0.
     \end{cases}
  \end{equation}
  
  Putting this together, we then have that for any $\mathbf{q} \in H_{F(p_1,p_2)}$,
  \begin{equation}
    \sum_{(\alpha,U)\in \EE_{\mathbf{p}\rightarrow \mathbf{q}}} \alpha 
      = \sum_{c\in \clif_{\mathbf{p}\rightarrow \mathbf{q}}} \frac{1}{|\clif_n^d|} 
      = \frac{|\clif_{\mathbf{p}\rightarrow \mathbf{q}}|}{|\clif_n^d|}
      = \frac{1}{|H_{F(p_1,p_2)}|}
      = \begin{cases} \frac{1}{d(d^{2n}-1)(d^{2n} - 2d)} & F(p_1,p_2) = 0\\
        \frac{1}{d^{2n+1}(d^{2n}-1)} & F(p_1,p_2) \neq 0.
     \end{cases}
  \end{equation}
  and the generalized Clifford group is Pauli 2-mixing.
\end{proof}


\section{Pauli 2-mixing on qubits implies 3-design\label{sec:proof}}

The proof proceeds by leveraging our knowledge of the action of $T_3$ on some well chosen $X$, by means of \propo{unit_twirl_preserves_permutation} and the linearity of $T_3$.  In particular, if $\EE$ is the ensemble of interest, we will show that for any $X\in \paul_n^{\otimes 3}$, $\Psi_{\EE,3}(X)$ is contained within the span of the permutation operators.  We can then use \propo{unitary_twirl_is_projective} to see that $\Psi_{\EE,3}(X) = T_3\big(\Psi_{\EE,3}(X)\big) = T_3(X)$.

The proof proceeds via case analysis, since the action of  $\Psi_{\EE,3}$ on $p_1\otimes p_2\otimes p_3$ depends on the relation between the three operators $p_1$, $p_2$, and $p_3$.  Note that we will only be working with qubit systems in this section, as these results do not hold for qudits.
\begin{lemma}
  If $\EE$ is a Pauli-invariant and Pauli 2-mixing $\clif_n$-ensemble, then $\EE$ is a unitary 3-design.
\label{lem:two_mixing_implies_3_design}
\end{lemma}

\begin{proof}
Let us define the channel $\Phi_k(Y) := \Psi_{\EE,k}(Y)$ for all $Y\in L\big(\XX^{\otimes k}\big)$, and use the shorthand $\Phi := \Phi_3$ and $T := T_3$.  We need to show that for all $Y\in L\big(\mathcal{X}^{\otimes 3}\big)$, $\Phi(Y) = T(Y)$.  Using \propo{unitary_twirl_is_projective}, it suffices to show that $\Phi(Y) = T\big(\Phi(Y)\big)$, and by \propo{unit_twirl_preserves_permutation} showing that $\Phi(Y) = \sum_{\pi\in S_3}\alpha_\pi W_\pi$ is enough.  The linearity of the two maps ensures that we need only show this for $X \in \paul_n^{\otimes 3}$, as the Pauli group forms a basis for $L\big(\XX^{\otimes 3}\big)$.

Assuming that $X = p_1\otimes p_2 \otimes p_3$ for $p_i\in \paul_n$, there are three cases (up to permutations of the underlying registers) that we need to consider, corresponding to the different relations between the three $p_i$.


\begin{case}[$X = p_1\otimes p_2\otimes \II_{\XX}$]

Let us first assume that at least one of the $p_i$ is the identity operator on $\XX$.  Without loss of generality, we can assume that $p_3 = \II_{\mathcal{X}}$ and
\begin{align}
  \Phi(X) &=\sum_{(\alpha,U)\in \EE} \alpha U p_1 U^\dag \otimes Up_2 U^\dag \otimes U( \II_{\XX}) U^\dag
      =  \sum_{(\alpha,U)\in \EE} \alpha Up_1 U^\dag \otimes Up_2 U^\dag \otimes \II_{\XX}
      = \Phi_2(p_1 \otimes p_2) \otimes \II_{\XX}\\
    &= T_2(p_1 \otimes p_2) \otimes \II_{\XX}
      = \int d\eta(U) Up_1U^\dag \otimes U p_2 U^\dag \otimes UU^\dag
      = T(X),
\end{align}
where we have used \lem{pauli_mixing_2_design} and \lem{2_mixing_implies_mixing} to show that $\EE$ is a 2-design.
\end{case}


\begin{case}[$X = p_1 \otimes p_2 \otimes p_1p_2$]\label{case:cycle}

Now let us assume that $X = p_1\otimes p_2 \otimes p_1p_2$, with $p_1,p_2 \in \widehat{\paul}_n$ and $p_1\neq p_2$.  We note that
\begin{align}
  \Phi(X) &= \sum_{(\alpha,U)\in \EE} \alpha U p_1U^\dag \otimes U p_2 U^\dag \otimes U p_1 U^\dag U p_2 U^\dag\\
    &= \sum_{(q_1,q_2)\in H_{F(p_1,p_2)}} \sum_{(\alpha,U)\in \EE_{(p_1,p_2)\rightarrow (q_1,q_2)}} \alpha U p_1U^\dag \otimes U p_2 U^\dag \otimes U p_1 U^\dag U p_2 U^\dag\\
    &= \sum_{(q_1,q_2)\in H_{F(p_1,p_2)}} \frac{1}{|H_{F(p_1,p_2)}|} \;q_1 \otimes q_2 \otimes q_1q_2\\
    &= \frac{4}{|H_{F(p_1,p_2)}|} \sum_{\substack{q_1\neq q_2\in \widehat{\paul}_n\\
    								F(p_1,p_2) = F(q_1,q_2)}} q_1\otimes q_2 \otimes q_1q_2,
\end{align}
where we used the fact that Clifford operators can only map $(p_1,p_2)$ to elements of $H_{F(p_1,p_2)}$ in the second equality, the fact that $\EE$ is Pauli 2-mixing in the third equality, and the fact that we only ever use $q_1$ and $q_2$ an even number of times so that $\pm q_i$ are equivalent in the last equality.

At this point we should recall the expansion of the permutation operators $W_{\pi}$ in the Pauli basis:
\begin{equation}
  W_{(321)} = \frac{1}{2^{2n}}\sum_{q_1,q_2\in \paul_n} q_1\otimes q_2 \otimes q_1q_2,
\end{equation}
and
\begin{equation}
   W_{(123)} = \frac{1}{2^{2n}} \sum_{q_1,q_2\in\paul_n} q_1 \otimes q_2 \otimes q_2 q_1 = \frac{1}{2^{2n}} \sum_{\substack{q_1,q_2\in \paul_n\\F(q_1,q_2) = 0}} q_1\otimes q_2 \otimes q_1q_2 - \frac{1}{2^{2n}} \sum_{\substack{q_1,q_2\in \paul_n\\F(q_1,q_2) = 1}} q_1\otimes q_2\otimes q_1q_2.
\end{equation}

Combining these two operators (along with some additional permutation operators to remove terms with $q_i = \II_\XX$), we can see that
\begin{equation}
  \Phi(X) = \frac{4}{|H_{F(p_1,p_2)}|}\begin{cases}  2^{2n-1} \big(W_{(123)} + W_{(321)}\big) -2^n \big(W_{(12)} + W_{(23)} + W_{(13)}\big) + 2 \II_{\XX^{\otimes 3}} & F(p_1,p_2) = 0\\
  2^{2n-1} \big(W_{(321)} - W_{(123)}\big) & F(p_1,p_2) = 1.
  \end{cases}
\end{equation}
Noting that $\Phi(X)$ is then contained within the span of the permutation operators, we can then use \propo{unit_twirl_preserves_permutation} and the linearity of $T$ to show that $\Phi(X) = T\big(\Phi(X)\big)$.  If we then use \propo{unitary_twirl_is_projective}, we have that $\Phi(X) = T(X)$.
\end{case}


\begin{case}[$X = p_1\otimes p_2 \otimes p_3$]
Now let us assume that $X = p_1 \otimes p_2 \otimes p_3$ with $p_1p_2p_3\not\propto \II_\XX$ so that there exists some $r\in \widehat{\paul}_n$ satisfying $F(r,p_1p_2p_3) = 1$.  This implies that $r$ anti-commutes with either one or three of the $p_i$, which then gives
\begin{align}
  \Phi(X) &= \sum_{(\alpha,U)\in \EE} \alpha Up_1U^\dag \otimes U p_2 U^\dag \otimes U p_3 U^\dag\\
    &= \frac{1}{2} \sum_{(\alpha, U)\in \EE} \alpha \Big(Up_1U^\dag \otimes U p_2 U^\dag \otimes U p_3 U^\dag +  Urp_1rU^\dag \otimes Ur p_2r U^\dag \otimes Ur p_3r U^\dag\Big)\\
    &= \frac{1}{2}\sum_{(\alpha,U)\in\EE} \alpha Up_1U^\dag \otimes U p_2 U^\dag \otimes U p_3 U^\dag \big(1 + (-1)^{F(r,p_1) + F(r,p_2) + F(r,p_3)}\big) = 0,
\end{align} 
where we used the fact that $\EE$ is Pauli-invariant to split the sum in the second equality, and the fact that $F(r,p_1) +  F(r,p_2) + F(r,p_3) = 1 \Mod{2}$ to recombine it in the third equality.  Using the linearity of $T$ and \propo{unitary_twirl_is_projective} shows that $T(X) = \Phi(X)$.
\end{case}

As any Pauli element over $\XX^{\otimes 3}$ can be placed into one of the above three cases (after a suitable permutation of the underlying spaces), we have that $\Phi(X) = T(X)$ for all $X\in \paul_n^{\otimes 3}$, and $\EE$ is a 3-design. 
\end{proof}

As an immediate consequence using \lem{cliffords_are_2_mixing}, we have our titular theorem:

\begin{theorem}
The Clifford group forms a unitary 3-design.
\label{thm:cliffords_are_3_design}
\end{theorem}


\section{The Clifford group is not a 4-design\label{sec:not_4_design}}

With the rather surprising result that the uniform distribution over the Clifford group forms a 3-design, one might wonder whether it also forms a $k$-design for larger $k$.  We can immediately see that this is false on one qubit, based on the size of $\clif_1$ and the minimal size of a 4-design.  However, this argument does not generalize to $n$ qubits for $n>1$ and leaves open the possibility that as $n$ grows, the uniform distribution over Clifford group becomes a better approximation to Haar-random unitaries.

We close this possibility, and show that the Clifford group only forms a 3-design.  We actually show a stronger statement, in that any 4-design requires the use of a non-Clifford unitary.  The main idea behind this proof is that any ensemble of Clifford elements only takes Pauli elements to Pauli elements.  If we examine the action of the 4-fold Clifford twirl on an operator of the form $p^{\otimes 4}$, then the image will only have support on Pauli elements of the form $q^{\otimes 4}$.  By showing that this subspace is not contained within the span of the permutation operators, we can use \propo{unit_twirl_spanned_by_permutation} to claim that the ensemble is not a 4-design.

Note that we again assume that we are working with qubits.  
\begin{lemma}
  If $\EE$ is a $\clif_n$-ensemble, then $\EE$ is not a 4-design.
\end{lemma}
\begin{proof}
  Let $p\in \widehat{\paul}_n$ and let us examine the action of $\Psi_{\EE,4}$ on $X = p^{\otimes 4}$.  We have that
  \begin{equation}
    \Psi_{\EE,4} (X) = \sum_{(\alpha,U)\in \EE} \alpha \big(U pU^\dag\big)^{\otimes 4} = \sum_{q\in \widehat{\paul}_n} \beta_q q^{\otimes 4} \qquad\text{ where} \qquad \beta_q = \sum_{(\alpha,U)\in \EE_{p\rightarrow \pm q}} \alpha .
  \end{equation}
  As the $\alpha$ form a probability vector and as each term in the sum uses $q_i$ an even number of times, we have that each $\beta_q$ is a sum of nonnegative terms and thus the $\beta_q$ also form a probability vector.  In particular, there must exist some $r_0\in \widehat{\paul}_n$ such that $\beta_{r_0} \neq 0$.  Additionally, we have that for any $Y\in \paul_{4n}$, $\big\langle Y, \Psi_{\EE,4}(X)\big\rangle = 0$ unless $Y = q^{\otimes 4}$ for some $q\in \widehat{\paul}_n$.  This is the key feature that will allow us to show $\Psi_{\EE,4}(X)\neq T_4(X)$.
  
  Let us now examine $T_4(X)$.  Using \propo{unit_twirl_spanned_by_permutation}, we have that
  \begin{equation}
    T_4(X) = \sum_{\pi\in S_4} \alpha_{\pi}(X) W_\pi.
  \end{equation} 
  Making use of the fact that for each $\pi\in S_4$ and each $U\in \unitary\big(\CC^{2^n}\big)$, the operator $W_\pi$ satisfies $[W_\pi,X] = 0$ and $[W_\pi,U^{\otimes 4}] = 0$, we can see that 
  \begin{equation}
    T_4(X) = T_4\big(W_\pi X W_\pi^\dag\big) = W_\pi T_4(X) W_\pi^\dag
  \end{equation}
  for all $\pi\in S_4$.  This shows that $\alpha_\pi(X)$ depends only on the conjugacy class of $\pi$, so let us label these coefficients by the length of their non-identity cycles (e.g., $\alpha_4$, $\alpha_{2,2}$, $\alpha_3$, $\alpha_{2}$, and $\alpha_{\emptyset}$).
  
  With this knowledge, let us examine the sum of $W_\pi$ for the order-4 permutations:
  \begin{align}
    \sum_{\pi = (abc4)} W_\pi &= \frac{1}{2^{3n}}\sum_{q_1,q_2,q_3\in \paul_n} q_1\otimes q_2\otimes q_3\otimes \big( q_1q_2q_3 + q_1 q_3 q_2 + q_2 q_1q_3 + q_2q_3q_1 + q_3 q_1q_2 + q_3q_2q_1\big)\\
    &= \frac{6}{2^{3n}} \sum_{\substack{q_1,q_2,q_3\in \paul_n\\ F(q_i,q_j) = 0}} q_1\otimes q_2\otimes q_3 \otimes q_1q_2q_3 
       -\frac{2}{2^{3n}}\sum_{\substack{q_1,q_2,q_3\in \paul_n\\
		\sum_{i < j}F(q_i,q_j) = 2}} (-1)^{F( q_1,q_3)}q_1\otimes q_2 \otimes q_3 \otimes q_1q_2 q_3      .
  \end{align}
  Similarly, if we look at the sum of $W_\pi$ for $\pi$ a pair of commuting cycles, we have that
  \begin{align}
    \sum_{\pi = (ab)(cd)} W_\pi &= \frac{1}{2^{2n}} \sum_{q_1,q_2\in \paul_n} q_1\otimes q_1 \otimes q_2 \otimes q_2 + q_1\otimes q_2\otimes q_1 \otimes q_2  + q_1 \otimes q_2 \otimes q_2 \otimes q_1.
  \end{align}
  
  Let us note that for any $Y\in \widehat{\paul}_{n}^{\otimes 4}$ and any $\pi\in S_4$, if $\big\langle Y,W_\pi\big\rangle \neq 0$, then $\pi$ must either be rank-4 or a pair of commuting cycles, as every other permutation has a fixed point.  (The fixed point forces the decomposition of $W_\pi$ in the Pauli basis to have one of the $q_i$ always equal the identity.)
  
  Now let us assume that $\Psi_{\EE,4}(X) = T_4(X)$.  Under this assumption, at least one of $\alpha_4$ or $\alpha_{2,2}$ must be nonzero, as there exists some $r_0\in \widehat{\paul}_n$ such that 
  \begin{equation}
   \big\langle r_0^{\otimes 4} , T_{4}(X)\big\rangle = \frac{6}{2^{3n}} \alpha_4 + \frac{3}{2^{2n}} \alpha_{2,2} =  \big\langle r_0^{\otimes 4} , \Psi_{\EE,4}(X)\big\rangle =  \beta_{r_0} \neq 0.
  \end{equation}
  
  Now let $r_1,r_2\in \widehat{\paul}$ satisfy $F(r_1,r_2)=1$.  Note that
  \begin{align}
    \Big\langle r_1\otimes r_1\otimes r_2\otimes r_2, \sum_{\pi = (abc 4)} W_\pi\Big \rangle &= \frac{2}{2^{3n}}, 
    & \Big\langle r_1\otimes r_2\otimes r_1\otimes r_2, \sum_{\pi = (abc 4)} W_\pi\Big \rangle &= - \frac{2}{2^{3n}},\\
    \Big\langle r_1\otimes r_1\otimes r_2\otimes r_2, \sum_{\pi = (ab)(cd)} W_\pi\Big \rangle &= \frac{1}{2^{2n}},
    & \Big\langle r_1\otimes r_2\otimes r_1\otimes r_2, \sum_{\pi = (ab)(cd)} W_\pi\Big \rangle &= \frac{1}{2^{2n}}.
  \end{align}
  With these equalities, we can see
  \begin{align}
    \big\langle r_1\otimes r_1\otimes r_2\otimes r_2, T_4(X)\big\rangle&= \frac{2}{2^{3n}} \alpha_{4} + \frac{1}{2^{2n}} \alpha_{2,2},
     & \big\langle r_1\otimes r_2\otimes r_1\otimes r_2, T_4(X)\big\rangle &= -\frac{2}{2^{3n}} \alpha_{4} + \frac{1}{2^{2n}} \alpha_{2,2},
  \end{align}
  while
  \begin{align}
    \big\langle r_1\otimes r_1\otimes r_2\otimes r_2, \Psi_{\EE,4}(X)\big\rangle&= 0,
     & \big\langle r_1\otimes r_2\otimes r_1\otimes r_2, \Psi_{\EE,4}(X)\big\rangle &= 0.
  \end{align}
  If we assume that $\Psi_{\EE,4}(X) = T_4(X)$, then $\alpha_4 = 0$ and $\alpha_{2,2} = 0$.  However, we have already shown that this assumption requires least one of these two coefficients to be nonzero.  Hence, our assumption leads to a contradiction, and $\EE$ is not a 4-design.  
\end{proof}

\begin{theorem}
  The Clifford group is not a unitary 4-design.
\end{theorem}


\section{The generalized Clifford group is not a 3-design\label{sec:gen_Pauli_not_design}}

After becoming aware of the work of Zhu \cite{Zhu15}, the author wanted to know where the proof of \lem{two_mixing_implies_3_design} fails for generalized Clifford operators, and whether the proof technique using the $k$-fold twirl channels could be used as an alternative proof of Zhu's results.  In particular, the author was interested in where the proof used the fact that the operations were on qubits as opposed to qudits.  

The key place we used this fact was \cas{cycle} of \lem{two_mixing_implies_3_design}, where we used the two order-3 permutation operators to cancel the commuting and anti-commuting terms of the form $q_1\otimes q_2 \otimes q_1q_2$.  When we generalize this result to qudits, there are $d$ different commutation possibilities as opposed to two, and the proof fails as we cannot span a $d$-dimensional space using only two vectors (unless $d=2$).

\begin{lemma}
If $\EE$ is a $\clif_n^d$-ensemble for $d$ not a power of 2, then $\EE$ is not a 3-design.
\label{lem:generalized_2_mixing_not_3_design}
\end{lemma}
\begin{proof}
  Let us examine the action of $\Phi_{\EE,3}$ on $X = p_1\otimes p_2\otimes p_2^\dag p_1^\dag$, where $F(p_1,p_2) = 1$.  More concretely, we have that $p_1 p_2 = \omega p_2 p_1$, where $\omega$ is the $d$-th primitive root of unity.
  
  We can use the fact that $\EE$ is a $\clif_n^d$-ensemble to see that
  \begin{equation}
    \Phi_{\EE,3}(X) = \sum_{\mathbf{q}\in H_1} \beta_{\mathbf{q}}\; q_1\otimes q_2\otimes q_2^\dag q_1^\dag \qquad \text{ where } \qquad \beta_{\mathbf{q}} = \sum_{(\alpha,U)\in \EE_{\mathbf{p}\rightarrow \mathbf{q}}} \alpha,
  \end{equation}
  as generalized Clifford elements can only map $(p_1,p_2)$ to elements of $H_1$.  Further, since both $q_i$ and $q_i^\dag$ are used an equal number of times in the decomposition, the phase of the $q_i$ does not matter in the expansion and we have
  \begin{equation}
    \Phi_{\EE,3}(X) = \sum_{\substack{\mathbf{q}\in (\widehat{\paul}_n^d)^2\\
      F(q_1,q_2) = 1}} \gamma_{\mathbf{q}}\; q_1\otimes q_2\otimes q_2^\dag q_1^\dag \qquad \text{ where} \qquad \gamma_{\mathbf{q}} = \sum_{i,j=0}^{d-1} \beta_{(\omega^i q_1,\omega^j q_2)} .
  \end{equation}
  Note that the $\gamma_{\mathbf{q}}$ are each a nonnegative sum of $\alpha$ terms, and that each $\alpha$ appears in exactly one decomposition of the $\gamma_{(q_1,q_2)}$.  As the $\alpha$ form a probability vector, at least one of the $\gamma_{\mathbf{q}}>0$ and $\Phi_{\EE,3}(X) \neq 0$.  We can also see that if $\big\langle Y, \Phi_{\EE,3}(X)\big\rangle\neq 0$ then $Y = q_1\otimes q_2\otimes q_2^\dag q_1^\dag$, with $F(q_1,q_2) = 1$. 
  
  If we now examine $T_3$, we can use \propo{unit_twirl_spanned_by_permutation} to see that
  \begin{equation}
    T_3(X) = \sum_{\pi\in S_3} \alpha_\pi(X) W_\pi,
  \end{equation}
  for some $\alpha \in \CC^{S_3}$.  Let us recall the decomposition of the order-3 permutation operators in the generalized Pauli basis:
  \begin{align}
    W_{(123)} &= \frac{1}{d^{2n}}\sum_{q_1,q_2\in \paul_n^d} q_1 \otimes q_2 \otimes q_2^\dag q_1^\dag, \\
    W_{(321)} &= \frac{1}{d^{2n}}\sum_{q_1,q_2\in\paul_n^d} q_1 \otimes q_2 \otimes q_1^\dag q_2^\dag = \frac{1}{d^{2n}} \sum_{\ell = 0}^{d-1} \omega^\ell\sum_{\substack{q_1,q_2\in \paul_n^d\\ F(q_1,q_2) = \ell}} q_1 \otimes q_2 \otimes q_2^\dag q_1^\dag.
  \end{align}
  
  It will be useful to note that if $Y\in \big(\widehat{\paul}_n^d\big)^{\otimes 3}$ and $\big\langle Y, W_\pi\big\rangle \neq 0$, then $\pi$ must be a rank-3 permutation, as every other permutation of three elements has a fixed point.  With this fact and our characterization of $\Phi_{\EE,3}(X)$ and $T_3(X)$, we have that if $\Phi_{\EE,3}(X) = T_3(X)$ then at least one of $\alpha_{(123)}$ or $\alpha_{(321)}$ must be nonzero.  
    
  Now, let $r_1,r_2\in \widehat{\paul}_n^d$ satisfy $r_1^2\not\propto \II$ and $F(r_1,r_2) = 2$.  By inspecting the above decompositions, we can see 
  \begin{align}
    \big\langle r_1\otimes r_1\otimes \big(r_1^2\big)^\dag, W_{(123)}\big \rangle &= \frac{1}{d^{2n}},
        &\big\langle r_1\otimes r_2\otimes r_2^\dag r_1^\dag, W_{(123)}\big \rangle &= \frac{1}{d^{2n}},\\
    \big\langle r_1\otimes r_1\otimes \big(r_1^2\big)^\dag, W_{(321)}\big \rangle &= \frac{1}{d^{2n}},
       & \big\langle r_1\otimes r_2\otimes r_2^\dag r_1^\dag, W_{(321)}\big \rangle &= \frac{\omega^2}{d^{2n}}.
  \end{align}
  Using the fact that $W_{(123)}$ and $W_{(321)}$ are the only permutation operators having nonzero support over $\big(\widehat{\paul}_n^d\big)^{\otimes 3}$, we can then examine the channels:
  \begin{align}
    \big\langle r_1\otimes r_1\otimes \big(r_1^2\big)^\dag, T_3(X)\big \rangle &= \frac{1}{d^{2n}}\big(\alpha_{(123)} + \alpha_{(321)}\big) ,&\big\langle r_1\otimes r_2\otimes r_2^\dag r_1^\dag, T_3(X)\big \rangle &= \frac{1}{d^{2n}}\big(\alpha_{(123)} + \omega^2 \alpha_{(321)}\big),\\
    \big\langle r_1\otimes r_1\otimes \big(r_1^2\big)^\dag, \Psi_{\EE,3}(X)\big \rangle &= 0, &\big\langle r_1\otimes r_2\otimes r_2^\dag r_1^\dag, \Psi_{\EE,3}(X)\big \rangle &= 0.
  \end{align}
  Hence, if $T_3(X) = \Psi_{\EE,3}(X)$, we require $\alpha_{(123)} = \alpha_{(321)} = 0$.  However, we have already shown that this assumption requires at least one of these two coefficients to be nonzero.  Hence the two channels cannot be equal and $\EE$ is not a 3-design.
\end{proof}

\begin{theorem}
The generalized Clifford group on qudits for $d$ not a power of $2$ is not a unitary 3-design.
\end{theorem}


\section{Discussion\label{sec:discussion}}

Our results explain why the Clifford group has been such a good approximation to the Haar-random unitaries: it was actually a better approximation than previously known.  These results might also explain why finding exact $k$-designs has been difficult; the prototypical example of a 2-design was actually a 3-design.  Hence, any attempt to generalize the Clifford group to find $k$-designs was attempting to extrapolate from the wrong data.

One possible way to improve these results is to find a Pauli 2-mixing ensemble other than the uniform ensemble over the Clifford group.  It is known that the size of the Clifford group grows like $2^{\Theta(n^2)}$, while the lower bounds on the size of a 3-design grows like $2^{\Omega(n)}$.  This large discrepancy in size gives weight to the conjecture that the uniform Clifford ensemble is not optimal; a 3-design should not require as large a set as the Clifford group.  However, \cite{Zhu15} has shown through some clever group theory manipulations that no proper subgroup of the Clifford group satisfies Pauli 2-mixing (except on two qubits).  Hence, no smaller Clifford ensemble can be the uniform distribution over a group.  However, this leaves open the possibility that the uniform distribution on some proper subset of the Clifford group still satisfies Pauli 2-mixing.

On the other hand, the results showing that the generalized Clifford group does not form a 3-design and that the Clifford group does not form a 4-design lends credence to the idea that we should abandon sending Pauli elements to Pauli elements.  If we wish to find $k$-designs for arbitrary $k$, we might first look for 2- and 3-designs that utilize unitaries outside the Clifford group.  Once these small designs are found, we might be able to generalize them to find $k$-designs for arbitrary $k$.

Finally, an application of an exact 3-design would be useful. Our results do give a simple proof that the stabilizer states form a state 3-design since they are simply the orbit of a pure state under the Clifford group, but this was already shown in \cite{KG13}.  The only other application of unitary 3-designs is shown in \cite{BH13}, in which they proved that a large fraction of approximate 3-designs are $\beta n$-dispersing (and thus so too are a large fraction of the Clifford group).  Finding an application for $k$-designs would thus be an interesting avenue of research.

\vspace{-.25cm}
\section*{Acknowledgments}

The author would like to thank Olivia Di Matteo for many helpful discussions and input on the early versions of this manusrcipt, and Huangjun Zhu for sharing a preprint copy of his result.

This work was supported by the Ontario Ministry of Training, Colleges, and Universities.


\bibliographystyle{myhamsplain}
\bibliography{cliffords3design}

\appendix
\section{Pauli-Mixing implies 2-design\label{app:pauli_mixing_2_design}}
We will prove \lem{pauli_mixing_2_design} in this appendix.  The proof of this lemma has a similar spirit to our proof of \lem{two_mixing_implies_3_design}.  Namely, we examine the action of the 2-fold twirl by the ensemble in question, and show that its image is contained within the span of the permutation operators.  Using  \propo{unitary_twirl_is_projective}, \propo{unit_twirl_preserves_permutation}, and the linearity of $T_2$, we can then claim that two channels have the same action on every Pauli element, which shows that the ensemble is a 2-design.

\pauliMixingTwoDesign*

\begin{proof}
  Let $\XX = \CC^{d^n}$, and let us define the channel $\Phi(Y) := \Psi_{\EE,2}(Y)$ for all $Y\in L\big(\XX^{\otimes 2}\big)$.  We need to show that for all $Y\in L\big(\XX^{\otimes 2}\big)$, $T_2(Y) = \Phi(Y)$.  Using \propo{unitary_twirl_is_projective}, it suffices to show that $\Phi(Y) = T_2(Y)$, and by \propo{unit_twirl_preserves_permutation} it suffices to show that $\Phi(Y) = \alpha \II_{\XX^{\otimes 2}} + \beta W_{(12)}$.  The linearity of the two maps ensures that we need only show this for $X\in \big(\paul_{n}^d\big)^{\otimes 2}$, as the generalized Pauli group forms a basis for $L\big(\XX^{\otimes 2}\big)$.
  
  Assuming that $X = p_1\otimes p_2$ for $p_1,p_2\in \paul_n^d$, there are three cases that we need to consider, corresponding to the different relations between $p_1$ and $p_2$.
  
\begin{case}[$X = \II_{\XX^{\otimes 2}}$] Let us first assume that $X = \II_{\XX^{\otimes 2}}$.  Direct inspections shows $\Phi(X) = \II_{\XX^{\otimes 2}} = T_2(X)$.
\end{case}

\begin{case}[$X = p\otimes p^\dag$]  Let us next assume that $X = p\otimes p^\dag$ where $p\in \widehat{\paul}_n^d$.  With this assumption, we have:
\begin{align}
  \Phi(X) &= \sum_{(\alpha,U)\in \EE} \alpha U p U^\dag \otimes U p^\dag U^\dag 
    = \sum_{q\in \overline{\paul}_n^d} \sum_{(\alpha,U)\in \EE_{p\rightarrow q}} \alpha q\otimes q^\dag
            = \sum_{q\in \overline{\paul}_n^d} \frac{1}{d(d^{2n}-1)} \;q \otimes q^\dag\\
  &
    = \frac{1}{d^{2n} - 1}\sum_{q\in \widehat{\paul}_n^d} q\otimes q^\dag,
\end{align}
where the second equality follows from the fact that $p$ can be mapped to any $q\in \overline{\paul}_n^d$, the third from the fact that $\EE$ is Pauli mixing, and the fourth from the fact that $q$ and $q^\dag$ both appear in equal powers, so that the phase does not affect our result.  

If we now recall the expansion of $W_{(12)}$ in the generalized Pauli basis, we can see
\begin{equation}
  W_{(12)} = \frac{1}{d^n}\sum_{p\in \paul_n^d} p\otimes p^\dag 
  \qquad \text{ and thus }\qquad
  \Phi(X) = \frac{1}{d^{2n}-1} \big(d^n \;W_{(12)} - \II_{\XX^{\otimes 2}}\big).
\end{equation}
Using \propo{unitary_twirl_is_projective}, \propo{unit_twirl_preserves_permutation}, and the linearity of $T_2$, we have that $T_2(X) = \Phi(X)$.
\end{case}

\begin{case}[$X = p_1\otimes p_2$] Let us finally assume that $X = p_1\otimes p_2$ where $p_1p_2 \not\propto \II_\XX$.  With this assumption, there exists some $r\in \widehat{\paul}_n^d$ such that $F(r,p_1p_2) = 1 = F(r,p_1) + F(r,p_2) \Mod d$.  We can then see
  \begin{align}
    \Phi(X) &= \sum_{(\alpha,U)\in \EE} \alpha U p_1 U^\dag \otimes U p_2 U^\dag 
       = \frac{1}{d} \sum_{\ell = 1}^d \sum_{(\alpha,U)\in \EE} \alpha Ur^\ell p_1r^{-\ell} U^\dag \otimes Ur^\ell  p_2 r^{-\ell} U^\dag \\
       &= \frac{1}{d} \sum_{\ell = 1}^d \omega^\ell \sum_{(\alpha,U)\in \EE} \alpha Up_1 U^\dag \otimes U  p_2 U^\dag 
       =  \Phi(X) \sum_{\ell=1}^d \frac{\omega^\ell}{d} 
       = 0,
  \end{align}
where we used the fact that $\EE$ is Pauli-invariant to split the sum in the second equality, the fact that $F(r,p_1) + F(r,p_2) \Mod{ d}= 1$ to recombine the sum in the third equality, and the fact that $\omega$ is a $d$-th root of unitary in the fifth equality.  Using \propo{unitary_twirl_is_projective} and the linearity of $T_2$ then gives that $T_2(X) = \Phi(X)$.
\end{case}

As any $X\in \paul_{2n}^d$ can be placed into one of these three cases, we have that $T_2(X) = \Phi(X)$, and $\EE$ is a 2-design.
\end{proof}

\end{document}